\documentclass[]{article}

\title{On the Locality of the Lov\'asz Local Lemma }
\author{}
\usepackage{bbm}

\usepackage[T1]{fontenc}
\usepackage[utf8]{inputenc}
\usepackage{lmodern}
\usepackage{fullpage}
\usepackage{url}
\usepackage{amsmath,amsthm,amssymb}
\newtheorem{theorem}{Theorem}

\newtheorem{lemma}[theorem]{Lemma}

\newtheorem{definition}[theorem]{Definition}

\usepackage{xspace}

\usepackage{cleveref}

\usepackage{algorithm}
\usepackage{algpseudocode}

\newcommand{\Prob}[1]{\mathbf{Pr}\left[#1\right]}

\def\epsilon{\ensuremath{\varepsilon }}
\newcommand{\eps}{\ensuremath{\epsilon }}
\newcommand{\Vars}{\ensuremath{\mathcal V }}

\newcommand{\ents}{\ensuremath{\mathcal X }}

\newcommand{\COMMENTED}[1]{{}}
\newcommand{\hide}[1]{\COMMENTED{#1}}

\newcommand{\local}{\textsf{LOCAL}\xspace}
\newcommand{\LOCAL}{\local}

\newcommand{\congc}{\textsf{CONGESTED CLIQUE}\xspace}
\newcommand{\MPC}{\textsf{MPC}\xspace}
\newcommand{\lca}{\textsf{LCA}\xspace}
\newcommand{\lcas}{\textsf{LCA}s\xspace}
\newcommand{\vol}{\textsf{VOLUME}\xspace}

\author{Peter Davies-Peck\thanks{The author is supported by EPSRC New Investigator Award UKRI155: ``Distributed Lov\'asz Local Lemma''.}\\Durham University\\\url{peter.w.davies@durham.ac.uk}}

\date{}

\begin{document}

\begin{titlepage}
\maketitle
\thispagestyle{empty}
\begin{abstract}
	The Lov\'asz Local Lemma is a versatile result in probability theory, characterizing circumstances in which a collection of $n$ `bad events', each occurring with probability at most $p$ and dependent on a set of underlying random variables, can be avoided. It is a central tool of the probabilistic method, since it can be used to show that combinatorial objects satisfying some desirable properties must exist. While the original proof was existential, subsequent work has shown algorithms for the Lov\'asz Local Lemma: that is, in circumstances in which the lemma proves the existence of some object, these algorithms can constructively find such an object. One main strand of these algorithms, which began with Moser and Tardos's well-known result (JACM 2010), involves iteratively resampling the dependent variables of satisfied bad events until none remain satisfied.
	
	In this paper, we present a novel analysis that can be applied to resampling-style  Lov\'asz Local Lemma algorithms. This analysis shows that an output assignment for the dependent variables of most events can be determined only from $O(\log \log_{1/p} n)$-radius local neighborhoods, and that the events whose variables may still require resampling can be identified from these neighborhoods. This allows us to improve randomized complexities for the constructive  Lov\'asz Local Lemma (with polynomial criterion) in several parallel and distributed models. In particular, we obtain:
	
	\begin{itemize}
		\item A \textsf{LOCAL} algorithm with $O(\log\log_{1/p} n)$ node-averaged complexity (while matching the $O(\log_{1/p} n)$ worst-case complexity of Chung, Pettie, and Su).
		\item An algorithm for the \textsf{LCA} and \textsf{VOLUME} models requiring $d^{O(\log\log_{1/p} n)}$ probes per query.
		\item An $O(\log\log\log_{1/p} n)$-round algorithm for \textsf{CONGESTED CLIQUE}, linear space \textsf{MPC}, and Heterogenous \textsf{MPC}.
	\end{itemize}

\end{abstract}
\end{titlepage}

\section{Introduction}

In this paper we provide a new analysis of resampling-style algorithms for the  Lov\'asz Local Lemma that allows most events to have their dependent variables fixed quickly. This leads to improved algorithms in several distributed and parallel models.

\subsection{The Lov\'asz Local Lemma}
The Lov\'asz Local Lemma (LLL) is a fundamental result in probability theory, widely used to prove the existence of mathematical objects via the probabilistic method. It applies in the following setting:

Suppose we have a set $\mathcal V$ of independent random variables. Based on these random variables is a set $\mathcal X$ of $n$ `bad events', that we wish to avoid. Suppose also that each bad event occurs with probability at most $p<1$. Under what criteria is it possible to avoid all the bad events (i.e. when are all avoided with non-zero probability under sampling the underlying random variables)?

A simple union bound can show that, if $np<1$, then the probability of avoiding all bad events is at least $1-\sum_{E\in \mathcal X} \Prob{E} \ge 1-np >0$. But what if $np\gg 1$?

If the bad events were independent (i.e. depended on disjoint subsets of the variables), then we would avoid all bad events with probability $\prod_{E\in \mathcal X} (1- \Prob{E}) \ge (1-p)^n>0$. That is, we can avoid all bad events regardless of the values of $n$ and $p$. But, in many cases, it is unrealistic to expect that all of the events we are interested in are independent.

The Lov\'asz Local Lemma represents a compromise: it states that if the bad events are \emph{mostly} independent, then we can still get positive probability of avoiding them. Formally, assume that each bad event is dependent on (i.e. shares underlying random variables with) at most $d$ others. The LLL then states the following:

\begin{theorem} [Lov\'asz Local Lemma \cite{EL74,Shearer85}]\label{thm:lll}
	If $epd\le 1$, then there exists an assignment of the random variables that avoids all bad events.
\end{theorem}

Notice that this result does not have any dependence on $n$; instead, it requires only a criterion relating event probability $p$ and dependence degree $d$. Multiple criteria are studied; this particular one is due to Shearer \cite{Shearer85}. We state the \emph{symmetric} version of the LLL here; this is a special case of the \emph{asymmetric} version, in which the probability bound for each bad event is allowed to be different.

This result has proven a versatile and widely-used tool in probability theory, combinatorics, and algorithms, with applications in many areas including  routing and scheduling \cite{LMR94}, hash function families \cite{DSW04}, satisfiability of Boolean formulae \cite{GMSW09}, and integer programming \cite{CS00}. The original proof of the LLL was existential, but substantial work has been devoted to the \emph{constructive} LLL: to providing an algorithm that produces an assignment of the underlying random variables which indeed avoids all bad events (see, e.g., \cite{ Beck91,  CS00,PT09}). Notably, the celebrated result of Moser and Tardos \cite{MT10}, for which they received the 2020 Gödel Prize, provided an efficient LLL algorithm for almost all circumstances in which the lemma holds, and also gave the first parallel and (implicitly) distributed LLL algorithms.

\subsection{The Lov\'asz Local Lemma in Distributed and Parallel Models}
The structure of a Lov\'asz Local Lemma instance can be viewed as a \emph{dependency graph}: the bad events to be avoided are the nodes of the graph, and edges are drawn between events that are dependent (i.e. share underlying random variables). In this way, the LLL becomes a graph problem, and can be studied in graph-based models of parallel and distributed computing. Since we will be viewing the problem by considering the structure of its dependency graph, we will henceforth often refer to the bad events in our LLL instance as \emph{nodes}.

The LLL has proven central to the study of distributed algorithms. To see why, consider the scenario of a distributed graph problem on a graph with many nodes, but where the degrees of nodes are low (the number of nodes is generally denoted as $n$ and the maximum degree as $\Delta$).

The most efficient distributed algorithms for graph problems are generally randomized, but the analysis of fast randomized algorithms often runs into problems on these low-degree instances. A common occurrence is that one can show that an algorithm succeeds at any particular node in the graph (i.e., produce a locally-consistent output) with some failure probability $p$ that is small compared to the maximum degree $\Delta$, but not small enough compared to the number of nodes $n$ (e.g., $\frac{1}{\Delta^2}>p>\frac 1n$). This is insufficient to take a union bound to show that the algorithm succeeds on the entire input graph with good probability, and indeed we would expect that running the algorithm under most random seeds would cause some nodes to fail.

The running of this algorithm can be formulated as an LLL instance: the random variables in $\mathcal V$ are the random bits used by the algorithm to make its random choices, and the bad events in $\mathcal X$ are the events that the algorithm fails at any particular node. The dependency graph will therefore have a structure which is (or is similar to) a power graph of the original input graph. The dependency degree $d$ of the LLL instance will be some function of $\Delta$ (up to $\Delta^{2r}$ where $r$ is the running time of the algorithm, but possibly smaller). Then, if $epd\le 1$, the LLL shows that the algorithm succeeds everywhere with non-zero probability. This probability may still be too small to be useful when running the algorithm directly, but an LLL algorithm can be used to find a satisfying variable assignment. In essence, this means that distributed LLL algorithms act as meta-algorithms that can amplify the success probability of randomized algorithms for other problems. Based on this idea, Chang and Pettie \cite{CP19} showed that the distributed LLL plays a crucial role in distributed complexity theory. It has been used to show gaps in the distributed complexity hierarchy, and is the canonical problem for the important class of LCLs (locally-checkable labelling problems) with $\Theta(\log n)$ deterministic \local complexity and $poly \log \log n$ randomized \local complexity. 

The LLL has been recently studied in another distributed model, \textsf{CONGEST}, in which messages have a bandwidth restriction. Maus and Uitto \cite{MU21} show a $\log^{O(1)}\log n$-round \textsf{CONGEST} LLL algorithm for graphs with $d=O(1)$ (under some mild restrictions on the complexity of event specifications), while Halld{\'o}rsson, Maus, and Peltonen \cite{HMP24} give $\log^{O(1)}\log n$-round algorithms for some more restricted special cases of the LLL in general graphs, which are useful in applications such as $\Delta$-coloring \cite{HM24}.

The parallel LLL also has a long history of study, with a series of results showing that algorithms that are efficiently parallelizable in the \textsf{PRAM} model \cite{Alon91,MT10,HH17, Harris23}. The \textsf{PRAM} complexity of the LLL is now well-understood, but much less is known about the complexity on more modern parallel models, such as \MPC, that admit substantially more power per processor.

\subsection{Models and Prior Work}
In this section we introduce the distributed and parallel models in which we will show results, and give an overview of the LLL in those models.

\subsubsection{\LOCAL model}
The \LOCAL model, introduced by Linial \cite{Linial92}, is a simple distributed message-passing model. A graph $G=(V,E)$ is given as both the input graph and the communication graph of a network. The nodes of this graph are the processors, and aim to solve some graph problem on $G$. For deterministic algorithms, nodes are equipped with unique identifiers (IDs); for randomized algorithms, if IDs are required they can be generated uniformly at random from $[n^3]$ and will be unique with high probability. 

Nodes are allowed to perform any computation on information they possess (potentially even exponential-time computations). However, they are initially only aware of their own input, the values or estimates of relevant graph parameters (such as, in this case, $d$, $p$ and $n$), and their adjacent edges. To discover information from further in the input graph, nodes must communicate with their neighbors. 

Communication proceeds in synchronous rounds, in which nodes can send messages to their neighbors in $G$. No restriction is placed on messages; they can be as large as desired (and so, nodes may as well send all information they have to all neighbors every round). This means that the only limiting factor in \LOCAL is the locality of the region of the graph on which nodes can base their outputs: if a \LOCAL algorithm runs for $T$ rounds, then nodes must be able to give consistent outputs dependent only on their $T$-radius neighborhoods. After some number of communication rounds, nodes must terminate and provide their own local output (e.g. a color, or whether they are part of some output node set). The goal, as algorithm designers, is to minimize the number of communication rounds necessary for nodes to output an overall valid solution to some graph problem.

The usual measure of round complexity in the \local model is the first round in which all nodes in the graph have terminated. However, an alternative measure, recently introduced by Feuilloley \cite{F20}, is to take the average number of rounds until termination across all nodes. Barenboim and Tzur \cite{BY19} show improved node-averaged complexities for a variety of problems in low-arboricity graphs, while Balliu et al.\cite{BGKO22} provide lower bounds on the node-averaged complexity of maximal matching and matching independent set matching the well-known worst-case complexity lower bound of Kuhn, Moscibroda, and Wattenhofer \cite{KMW04}.

Moser and Tardos's seminal work \cite{MT10} gave the first parallel LLL algorithm, which also implicitly works as a distributed algorithm. This algorithm runs in $O(\log^2 n)$ rounds of randomized LOCAL for the LLL criterion $ep(d+1)< 1-\eps$ (for any constant $\eps>0$). The distributed LLL problem was formalized by Chung, Pettie and Su \cite{CPS17}, who, using a similar approach, presented an $O( \log^2 d\log_{1/ep(d+1)} n)$-round algorithm for the criterion $ep(d+1)< 1$ and an $O(\log_{1/epd^2} n)$-round algorithm for the criterion $epd^2<1$. The former result relied on computation of maximal independent sets (MIS), and its running time was subsequently improved to $O(\log d\log_{1/ep(d+1)} n)$ by the MIS result of Ghaffari \cite{Ghaffari16}. These resampling-style algorithms have node-averaged complexity asymptotically equal to their standard worst-case complexity, since their analysis did not allow any node to determine that it can safely output until the completion of the algorithm.

For cases when the dependency degree $d$ is low compared to $n$, Fischer and Ghaffari \cite{FG17} gave improved algorithms based on the approach of Molloy and Reed \cite{MR98}, later extended by Ghaffari, Harris, and Kuhn \cite{GHK18}. In light of the subsequent polylogarithmic network decomposition of Rozho\v{n} and Ghaffari \cite{RG20}, these works imply a distributed LLL algorithm requiring $O(d^2+\log^{O(1)}\log n)$ rounds, under the criterion $p< d^{-c}$ for sufficiently large constant $c$. Recently, this bound was improved to $O(\frac{d}{\log d}+\log^{O(1)}\log n)$ by Davies \cite{Davies23a}. These works do not explicitly give node-averaged round complexities, but it can be seen that the node-averaged complexity of \cite{Davies23a} is $O(\frac{d}{\log d})$. This is because the algorithm is a \emph{shattering}-style algorithm (a design technique for local algorithms that can be traced back to the LLL algorithm of Beck \cite{Beck91}). It consists of an $O(\frac{d}{\log d})$-round pre-shattering phase that determines the outputs for most nodes, and it can be shown that the probability that any particular node $v$ remains in an unsolved component of size $x$ is at most $2^{-x/9d}$. Then, the post-shattering stage produces an output for $v$ in $poly(\log x)$ rounds, and so the expected number of rounds until $v$ outputs is at most $O(\frac{d}{\log d} + \sum_{x=1}^{n} 2^{-x/9d} poly(\log x)) = O(\frac{d}{\log d})$. A more careful analysis can show that this node-averaged complexity bound applies with high probability, not just in expectation.

Therefore, for the LLL version we study (symmetric, with polynomial slack), the fastest algorithms are those of Chung, Pettie and Su \cite{CPS17} and Davies \cite{Davies23a}. This gives a worst-case (with high probability) round complexity of $O(\min\{\log_{1/p} n,\frac{d}{\log d}+ \log^{O(1)}\log n \})$, and a node-averaged complexity of $O(\min\{\log_{1/p} n,\frac{d}{\log d} \})$.

The Lov\'asz Local Lemma has proven important also for lower bounds, and was the first problem for which a lower bound was shown using the `round elimination' technique that has since produced a host of LOCAL lower bounds for various problems. The bound given, by Brandt et al.  \cite{BF+16}, was $\Omega(\log_{\log 1/p}  \log n)$ rounds, holding even for the much weaker criterion $p\le 2^{-d}$ (and this marks a sharp threshold in criterion strength, since Brandt, Grunau and Rozho\v{n} \cite{BGR20} demonstrate that weakening the criterion any further admits an $O(d^2+\log^* n)$-round deterministic algorithm). The lower bound of Brandt et al. holds even on trees, where it is essentially matched by an $O(\max\{\log_{\log 1/p} \log n, \log \log n/ \log \log \log n\})$ round algorithm \cite{CHLPU19}. However, on general graphs there remains a substantial gap between upper and lower bounds.

\subsubsection{\lca and \vol models}

The Local Computation Algorithms (\lca) model was introduced by Rubinfeld, Tamir, Vardi, and Xie \cite{RTVX11}. In the \lca model, the aim is for a centralized processor to answer queries about local parts of the solution to a graph problem - that is, when a node (or edge, depending on the problem) is queried, the processor must provide the local solution corresponding to that node (or edge). It does so by \emph{probing} nodes of the input graph, which we assume to have unique identifiers in $[n^{O(1)}]$ to be identifiable by the processor. Each probe of a node reveals any input of that node, as well as its adjacent edges. The processor must then provide an output for the queried node dependent only on the probed region of the graph (as well as a random string, for randomized \lcas).  The output given by the processor to all queries it receives must form part of a consistent, valid solution to the overall graph problem.

We will focus on the \emph{probe complexity} of \lca algorithms, i.e. the number of probes needed to respond to a query. Many works on \lca are also concerned with the processing time of the processor, but we cannot easily analyze this for the general LLL - it would depend heavily on the number of variables and the complexity of evaluating events, neither of which are constrained in the general LLL formulation. So, processing time is generally only analyzable for specific applications of the LLL.

The Lov\'asz Local Lemma is central to the model and has been well-studied: indeed, many works on \lca  since the model's inception (e.g. \cite{RTVX11,ARVX12,DK23}) have been based on analysis of the LLL, and studied problems that are LLL applications. Achlioptas, Gouleakis, and Iliopoulos \cite{AGI20} give an \lca algorithm for the general LLL, though this algorithm requires $n^{\Omega(1)}$ probes. Brandt,  Grunau, and Rozho{\v{n}} \cite{BGR21} show that the \lca complexity of the LLL in \emph{constant-degree graphs} (i.e. with $d=O(1)$) is $\Theta(\log n)$. \local model algorithms can be used in the \lca model, with \lca probe complexity $= d^{\text{\local complexity}}$, and so the \local algorithm of \cite{CPS17} implies a $d^{O(\log_{1/p} n)}$-probe \lca. The \local algorithm of Davies \cite{Davies23a} immediately implies a $d^{O(\frac{d}{\log d}+\log^{O(1)}\log n)}$-probe \lca, but again a slightly better bound can be obtained by considering the shattering nature of the algorithm: $d^{O(\frac{d}{\log d})}$ probes suffice to identify whether a node remains after the pre-shattering phase, and such nodes form components of size $\log n \cdot d^{O(1)}$. Therefore, a node can find its post-shattering component and determine its final output using $\log n \cdot d^{O(1)} \cdot d^{O(\frac{d}{\log d})} = 2^{O(d)}\log n$ probes. 

The \vol model, introduced by Rosenbaum and Suomela \cite{RS20}, is a similar model that aims to capture the size of the region that nodes in a graph need to explore in order to determine their output (as opposed to the radius of that region, as studied in the \local model). The primary differences from \lca are that the probed region is required to be connected (there is no such requirement in \lca), and randomness is local (i.e. each node has its own independent source of random bits, that are revealed when that node is probed, whereas in \lca shared randomness is assumed). These distinctions are often of little consequence, and that is the case here: the best \vol algorithms for the LLL are the same as the \lcas above.

\subsubsection{\congc, linear-space \MPC, and Heterogenous \MPC models}
\congc and \MPC are related parallel models, and allow all-to-all communication between processors (though \congc arose as a variant of the distributed \textsf{CONGEST} model). 

In \MPC, a group of machines is applied to solve a problem, with the input initially distributed arbitrarily over the machines. The primary restriction is that each machine has a space bound of $\mathfrak s$ words, and all input, output, and local computation must fit within this space bound. Algorithms in \MPC proceed in synchronous rounds, and in each round, machines may perform some computation on the information they have in their local memory, send up to $\mathfrak s$ words of information to other machines (divided however they choose), and receive up to $\mathfrak s$ words of information to other machines. The number of machines permitted is generally as few, or close to as few, as are needed to store the entire input $I$ (i.e. $\tilde O(\frac{|I|}{\mathfrak s})$ machines) and the goal is to minimize the number of communication rounds needed to solve the problem.

There are two main space regimes studied for graph problems: linear-space \MPC, in which $\mathfrak s = \tilde O(n)$ space, and sublinear-space (also called low-space) \MPC, in which $\mathfrak s = n^{\eps}$ for some constant $\eps\in (0,1)$. In this paper we study the former. In both linear-space and sublinear-space \MPC, many common `housekeeping' operations, such as searching, sorting, and computing prefix sums, are known to be computable in $O(1)$ rounds \cite{GSZ11}, and so round-complexity bottlenecks usually stem from the size of graph regions that must be collected to a machine to solve tasks.

The Heterogenous \MPC model, introduced recently by Fischer, Horowitz, and Oshman \cite{FHO22}, aims to demonstrate that only one machine need have linear space in order to achieve much of the power of linear-space \MPC. In this model, there is a single large machine with $\tilde O(n)$ space, and all other machines have space in the sublinear regime (i.e. $n^\eps$ for some $\eps \in (0,1)$). 

\congc does not have explicit space bounds for machines, but instead has a bandwidth bound: in each communication round, only $O(\log n)$ bits (i.e. $O(1)$ words) can be sent between each pair of machines. This may seem like a much more stringent requirement than that of \MPC; however, due to the constant-round routing of Lenzen \cite{L13}, messages can be routed to observe these bandwidth constraints, rendering \congc essentially equivalent to linear-space \MPC \cite{BDH18}.

The LLL is less well-studied in these models than in \local, \lca and \vol. To our knowledge, there is no LLL algorithm designed specifically for \congc, linear-space \MPC, or Heterogenous \MPC, and the best round complexity known in all three is $O(\min\{\log\log_{1/p} n,\log d\}) $. Here, the first term follows by simulating the algorithm of \cite{CPS17} or \cite{MT10} using graph exponentiation, and the second follows by simulating the randomized pre-shattering part of the algorithm of \cite{Davies23a} or \cite{FG17} using graph exponentiation, and then collecting the shattered graph onto a single machine to solve sequentially. Note that in the latter case, we may assume $d=o(\log n)$, since otherwise the first approach would be faster, and in that case the balls of radius $d^{\frac{d}{\log d}} = n^{o(1)}$ necessary to simulate the pre-shattering part of \cite{Davies23a} fit into machines' local memory.

However, the LLL is important in the study of lower bounds in low-space \MPC:  LLL-related problems are one of the classes for which component-instability is shown to help surpass conditional component-stable lower bounds (\cite{CDP24}, with the lower-bound framework revised from \cite{GKU19}).

\subsection{Our Approach}

In this section we give an overview of our approach.

\subsubsection{Version of the LLL}

The version of the LLL we will study is the symmetric LLL with \emph{polynomial criterion} (sometimes also known as polynomially-weakened \cite{GHK18,Davies23a}). That is, we study instances in which

\[p = O(d^{-c}),\]

for some constant $c > 1$. While weaker than the strongest criteria for which the LLL holds (such as that of Theorem \ref{thm:lll}), this criterion is still sufficient to capture many applications of the LLL, and has seen extensive study in previous sequential and distributed models (e.g. \cite{Beck91, FG17, Davies23a}).

Specifically, the criterion for which our results apply will be

\[p = d^{-(10+\Omega(1))}\]

We will also assume that $d=\omega(1)$ (i.e. increases with $n$), and therefore $p=o(1)$, since otherwise the \local LLL algorithms of Fischer and Ghaffari \cite{FG17} and Davies \cite{Davies23a} already imply the desired results. 

We are interested in randomized algorithms for the LLL that succeed \emph{with high probability (w.h.p.)}, i.e. with probability at least $1-n^{-c}$ for some constant $c\ge 1$. In our algorithms, this constant $c$ can be amplified to be arbitrarily large at only a constant-factor increase in complexity, so we do not specify it explicitly.

\subsubsection{LLL algorithm}

The starting point of our approach is the `simple algorithm' of Chung, Pettie, and Su \cite{CPS17} (Algorithm \ref{alg:CPSLLL}). Our analysis can also be applied to the parallel algorithm of Moser and Tardos \cite{MT10}, which would improve the exponent in the LLL criterion required (though it would still be polynomially weaker than Theorem \ref{thm:lll}). However, this would come at the expense of increased complexity in the models we study, since the algorithm requires repeated computation of a maximal independent set, which is comparatively costly. So, to optimize the final complexities we achieve, we restrict our focus here to the algorithm of \cite{CPS17}:

\begin{algorithm}[H]
	\caption{Chung, Pettie, and Su's algorithm}
	\label{alg:CPSLLL}
	\begin{algorithmic}
		\While{there are satisfied (bad) events}
		\State Let $\mathcal I$ be the set of satisfied events with locally minimal IDs (i.e. with no satisfied neighbor with a lower ID)
		\State Resample all dependent variables of events in $\mathcal I$.
		\EndWhile
	\end{algorithmic}
\end{algorithm}

It is shown that this algorithm terminates w.h.p. in $O(\log_{1/epd^2} n)$ iterations, so long as $epd^2<1$ \cite{CPS17}, meaning that for our criterion it requires $O(\log_{1/p} n)$ iterations. Furthermore, as in Moser and Tardos's algorithm, $O(np)$ events are resampled in total. So, most events retain their initially sampled variable values in the final output, and the average amount of `work' done per event is only $O(p)$.

However, in distributed and parallel models, the analysis of this style of algorithm had a major drawback: while most events never need to resample their variables, there was no known way to identify such `secure' events without running the entire algorithm. In a distributed setting, this means that all processors (representing the bad events in the LLL instance, and acting as nodes in the dependency/communication graph) must remain `active' and participating throughout the entire algorithm, even though most will eventually output their initial sampled variable values. The reason for this is that the dependency structure of the LLL instance can cause long `chains' of resampled events, resulting in events that appeared safe to eventually need to resample variables.

We provide a new analysis of the LLL that solves this problem. Specifically, we show that there exists a property that can be determined only from the $O(\log\log_{1/p} n)$-radius neighborhood of an event, that determines whether a node is at risk of being resampled later in the algorithm. In the \LOCAL model, this allows most nodes to identify that they can terminate and output their initial variable values within $O(\log\log_{1/p} n)$ rounds. In the \lca, \vol, and parallel models, it bounds the size of the region that is needed to determine the output values of an event's dependent variables.

\subsubsection{Structure of witness trees}

The main idea behind our results comes from examining the shape of \emph{witness trees}, structures used to analyze how resampling events' variables propagates through the LLL dependency graph. In particular, we focus on some event $v$, and examine whether the witness trees rooted at $v$ \emph{expand} away from $v$, by at least a constant factor every distance hop. That is, whether the number of nodes at a distance $i$ from $v$ is at least a constant factor higher than the number of nodes at distance less than $i$.

If witness trees do expand, then within radius $O(\log\log_{1/p} n)$ they must reach a size of at least $\omega(\log_{1/p} n)$ events, since the size is exponentially increasing in the radius. However, it can be shown that w.h.p. no witness trees of this size occur, and so this cannot be the case.

So, witness trees of $\omega(\log\log_{1/p} n)$ radius must fail to expand at at least one distance threshold $R$ (and we call such non-expanding trees truncated at this threshold \emph{narrow witness trees}).

The other major part of our analysis is to show that the probability of a narrow witness tree occurring is small \emph{even if behavior outside the radius $R$ is adversarial}. This is because this external behavior only affects the nodes at distance $R$ from $v$, which, by the definition of a narrow witness tree, constitute only a small fraction of the tree. So, the remainder of the tree (the nodes at distance less than $R$ from the root) is still sufficiently unlikely to occur.

Then, we show that the probability that any node is the root of such a narrow witness tree is small, and the induced graph of such roots \emph{shatters} into small pieces. Most of the resampling steps of Algorithm \ref{alg:CPSLLL}, therefore, can be performed on this smaller shattered graph, with the output for all other nodes able to be determined much earlier.

\section{Our Results}

Our algorithm returns exactly the same output as Algorithm \ref{alg:CPSLLL} (under the same randomness). What we show is a way in which most nodes can quickly detect whether they can retain their initial variable values for the final output. Our main technical result is the following:

\begin{lemma}[Insecure events]\label{lem:insecure}
Upon sampling a randomness table of variable values from their distributions, there is a property of events that we call \emph{insecure} satisfying the following:

\begin{itemize}
	\item Whether an event is insecure depends only on its $O(\log\log_{1/p} n)$-radius neighborhood in the dependency graph $G$.
	\item The dependent variables of secure (i.e. non-insecure) events are fixed (will not be resampled) after $O(\log\log_{1/p} n)$ resampling rounds of Algorithm \ref{alg:CPSLLL}, w.h.p.
	\item The induced graph $G'$ of insecure nodes has $np^{\Omega(\log\log_{1/p} n)} = o(n)$ nodes and edges, w.h.p.
	\item The connected components of $G'$ are of size $d^{O(\log\log_{1/p} n)}$ and diameter $O(\log_{1/p} n)$ (with respect to distances in $G$), w.h.p.
	\item $G'$ admits an $(O(\log\log_{1/p} n),O(\log^2\log_{1/p} n))$ network decomposition (again with respect to distances in $G$), w.h.p.
\end{itemize}
\end{lemma}

Based on this result, we improve the complexity of LLL in several models. Firstly, we show a \LOCAL algorithm that improves the node-averaged complexity exponentially, while retaining the worst-case complexity of \cite{CPS17}:

\begin{theorem}\label{thm:local}
	LLL can be solved in \LOCAL with  $O(\log\log_{1/p} n)$ node-averaged complexity (and $O(\log_{1/p} n)$ worst-case complexity).
\end{theorem}

Next, we show an algorithm for the \lca and \vol models with improved probe complexity per query:

\begin{theorem}\label{thm:lca}
	LLL can be solved in \lca and \vol in $d^{O(\log\log_{1/p} n)}$ probes per query.
\end{theorem}

This improves exponentially over the $d^{O(\log_{1/p} n)}$-probe \lca and \vol algorithm implied by Chung, Pettie, and Su's \local result \cite{CPS17}, and also improves substantially over the $2^{O(d)}\log n$-probe algorithm implied by Davies's \local result \cite{Davies23a} for most values of $d$.

Finally, we show an algorithm for the parallel models of \congc, linear-space \MPC, and Heterogenous \MPC, which also improves exponentially over the previous best round complexities:

\begin{theorem}\label{thm:mpc}
Constrained instances of LLL can be solved in \congc, linear-space \MPC, and Heterogenous \MPC in $O(\log \log \log_{1/p} n)$ rounds. The \MPC algorithms use $O(|I| + n^{1+o(1)})$ total space, where $I$ is the input.
\end{theorem}

Here, `constrained' instances are those that meet modest restrictions on the size and complexity of event specifications, since the space restrictions of the models mean that we cannot allow arbitrary events. However, to our knowledge, the definition we use (Definition \ref{def:constrained}) captures all distributed and parallel applications of the LLL.

\section{Analysis}

This section contains the technical arguments comprising the proof of Lemma \ref{lem:insecure}.

\subsection{Randomness Table and Execution Log}
We consider the `randomness table' analysis perspective to keep track of the values taken by random variables during the resampling process. Imagine a table in which each variable $x \in \Vars$ has an array of the values it takes upon each resampling: that is, its first sampled value is $x_0$, if resampled it takes value $x_1$, if resampled again it takes value $x_2$, and so on. Since it is known that Algorithm \ref{alg:CPSLLL} terminates in $O(\log_{1/p} n)$ resampling steps, we will need to consider values $x_0, x_1, \dots, x_{O(\log_{1/p} n)}$. Initially, all variables $x$ take their first sampled value $x_0$, and all subsequent values are `hidden' (by which we mean that, in the analysis, we will consider their values to be unfixed, and will analyze the probabilities of events dependent on them). When the algorithm requires the dependent variables of an event to be resampled, for each such variable the next value in the table is `revealed' and taken, and in the analysis is henceforth treated as fixed. So, upon analyzing any particular point during the algorithm, a variable $x$'s current value is $x_i$ where $i$ is the number of times $x$'s dependent events have been resampled, and values $x_0, \dots, x_i$ are revealed, with the remaining values still hidden.

The execution and output of Chung, Pettie, and Su's algorithm is uniquely determined by the randomness table. In each round of Algorithm \ref{alg:CPSLLL}, events which have locally minimal IDs in the induced graph of satsified events resample their dependent variables. By definition, the set of resampling events must be independent. Based on a fixed randomness table, we can then specify the particular independent set of events that will resample in each iteration:

\begin{definition}
	The \emph{execution log} of an execution of Chung, Pettie, and Su's algorithm is a sequence $S_{1}, S_{2},\dots$ of sets of events. $S_t$ is the (independent) set of events whose dependent variables are resampled in resampling step $t$.
\end{definition}

\subsection{Witness Trees}

We next define the concept of a witness tree derived from our input LLL dependency graph $G$:

\begin{definition}
A witness tree of $G$-radius $r$ is a directed tree $T$ rooted at a node $root(T)$ such that:

\begin{itemize}
	\item Nodes in $T$ are copies of nodes $v\in V(G)$. That is, the nodes of $T$ form a \emph{multiset}, and any node of $G$ can appear multiple times in $T$.
	\item All edges in $T$ are directed towards the root.
	\item $(u,v)\in T \implies dist_{G}(u,v)\le 2$.
	\item $\forall v \in T$, $dist_G(v,root(T))\le r$. That is, the nodes of $T$ are contained within the radius-$r$ ball around $root(T)$ in $G$ (but note that the height of $T$ can be larger than $r$). We call the events at nodes $v$ with $dist_G(v, root(T))  =r$ \emph{boundary events}, and denote the set of these events $\partial(T)$.
\end{itemize}

The \emph{size} of a witness tree $T$, which we will denote $|T|$, is its number of nodes. 
\end{definition}

Note that this definition coincides with the definition of $2$-witness trees used by Chung, Pettie, and Su in the analysis of Algorithm \ref{alg:CPSLLL} \cite{CPS17}, but with the added notion of $G$-radius, which we will need to analyze the structure of the witness trees that occur during an execution.

\begin{definition}
A witness tree $T$ \emph{occurs} for a given randomness table if it can be generated from the resulting execution log by the following process:

\begin{itemize}
	\item Choose some resampling step $t$ at which some $v$ is resampled (i.e. $v\in S_t$).
	\item Set $v$ to be the root of $T$.
	\item Proceed backwards through each resampling step $t-1,\dots,1$ in the execution log. 
	\item For each resampling step $i$, for each $u\in S_i$, if there is some $w$ in the tree $T$ with $dist_{G}(u,w)\le 2$, add $u$ as a child of the \emph{deepest} such $w$ (i.e. the furthest from the root $v$ of $T$, by distance in $T$ rather than $G$), breaking ties by highest ID.
\end{itemize}

\end{definition}

Since the nodes of $T$ form a multiset, we will sometimes denote them $v_{(t)}$, meaning the copy of event $v\in \ents$ resampled at time $t$, to distinguish multiple copies of the same event $v$.

Again, this definition of the occurrence of a witness tree during an execution coincides with that of \cite{CPS17}. Since all events resampled in a step are independent, every resampling step traced backwards through the execution log increases the depth of $T$ by at most $1$. Chung, Pettie, and Su show that it increases by exactly $1$, and therefore a witness tree with a root resampled in step $t$ is of depth exactly $t-1$. This gives the following lemma:

\begin{lemma}
	If an event $v$ is sampled in resampling step $t$ of Chung, Pettie, and Su's algorithm (i.e. is in $S_t$), then $v$ is the root of an occurring witness tree of size at least $t$, depth exactly $t-1$, and $G$-radius at most $2(t-1)$.
\end{lemma}

\begin{proof}
	By Lemma 3.2 of \cite{CPS17}, a witness tree with a root resampled in step $t$ is of depth exactly $t-1$. Then, the lemma follows since tree size is strictly greater than depth, and $G$-radius is at most twice depth (as adjacent nodes in $T$ are within distance $2$ in $G$).
\end{proof}

Witness trees are sufficient to determine the variable values taken during the course of the algorithm, as follows: 

\begin{lemma}\label{lem:varspec}
Let $T$ be an occurring witness tree, let $v_{(t)}\in V(T)$ be an event in $T$ resampled at resampling step $t$, and let $x\in \Vars$ be a dependent variable of $v$. Let $k$ be the number of dependent events of $x$ of depth (i.e. distance in $T$ from $root(T)$) at least that of $v_{(t)}$. Then, the value taken by variable $v$ after resampling step $t$ in Algorithm \ref{alg:CPSLLL} is $x_k$.
\end{lemma}

\begin{proof}
Any dependent event of $x$ that is resampled prior to resampling step $t$ will be present in $T$, since it is adjacent to $v$ and would therefore be added during the generation process, at a greater depth than $v_{(t)}$. Each of these resampled events, as well as $v_{(t)}$ itself, causes $x$ to take the next value in the randomness table, starting at $x_0$. So, after resampling step $t$, $x$ will take value $x_k$.
\end{proof}

\subsection{Possible Witness Trees}

In our first departure from the analysis of \cite{CPS17}, we next define a weaker notion of witness trees $T$ that are \emph{possible} in an execution. These are witness trees that are created from an execution log in the same manner as occurring witness trees, except that we ignore all nodes of distance more than $R$ from the root in $G$. Intuitively, the idea here is that if we only look at a radius of $R$ from this root, these trees are those that \emph{could} be occurring, depending on what happened outside this radius earlier in the execution log. 

\begin{definition}
	A witness tree $T$ is $R$-\emph{possible} in an execution if it can be generated from the execution log by the following process:
	
	\begin{itemize}
		\item Choose some resampling step $t$ at which some $v$ is resampled (i.e. $v\in S_t$).
		\item Set $v$ to be the root of $T$.
		\item Proceed backwards through each resampling step $t-1,\dots,1$ in the execution log. 
		\item For each resampling step $i$, for each $u\in S_i$ with $dist_{G}(u,v)\le R$, if there is some $w$ in the tree $T$ with $dist_{G}(u,w)\le 2$, add $u$ as a child of the \emph{deepest} such $w$ (i.e. the furthest from the root $v$ of $T$, by distance in $T$ rather than $G$), breaking ties by highest ID.
	\end{itemize}
	
\end{definition}

Clearly, if a witness tree $T$ of $G$-radius $r$ is $R$-possible, then $R\ge r$ (since otherwise the boundary events of $T$ would not have been added during the generation process). Furthermore, if $R>r$ then $T$ occurs (since no nodes $u$ with $dist_G(u,root(T))=R$ were added during generation, and therefore the same tree would have resulted from generation of an occurring witness tree). So, the notion of $R$-possible witness trees captures both occurring trees of $G$-radius less than $R$, and possible witness trees that could be part of occurring trees of $G$-radius at least $R$.

Possible witness trees similarly specify the values taken by variables during the course of the algorithm, except those at distance $R$ from the root:

\begin{lemma}\label{lem:varno}
	Let $T$ be an $R$-possible witness tree, let $v_{(t)}\in V(T)$ be an event in $T$ resampled at resampling step $t$, with $dist_G(v,root(T))<R$, and let $x\in \Vars$ be a dependent variable of $v$. Let $k$ be the number of dependent events of $x$ of depth (i.e. distance in $T$ from $root(T)$) at least that of $v_{(t)}$. Then, the value taken by variable $x$ after resampling step $t$ in Algorithm \ref{alg:CPSLLL} is $x_k$.
\end{lemma}

\begin{proof}
	As for Lemma \ref{lem:varspec}; any resampled event that would affect the variable values of $v$ at step $t$ is in $T$, at a greater depth than $v_{(t)}$. Each of these resampled events, as well as $v_{(t)}$ itself, causes $x$ to take the next value in the randomness table, starting at $x_0$. So, after resampling step $t$, $x$ will take value $x_k$.
\end{proof}

We can bound the probability that any particular witness tree occurs, or is possible.

\begin{lemma}\label{lem:treeprob}
The probability that any particular witness tree $T$ \emph{occurs} in an execution is at most $p^{|T|}$. The probability that any particular witness tree $T$ is $R$-\emph{possible} in an execution is at most $p^{|T|-|\partial(T)|}$. These probability bounds hold even if the sampled values of variables independent of the events in $T$ are chosen adversarially.
\end{lemma}

\begin{proof}
The argument for the occurrence probability is well-known and used in several LLL algorithms \cite{MT10,CPS17}, but we summarize it here: for a witness tree to occur, for any node $v_{(t)} \in V(T)$, $v$ must have been satisfied under the variable assignments at step $t$, or it would not have been reverted. Furthermore, the events that each such bad event in $V(T)$ is satisfied are independent, because they are dependent on disjoint entries in the randomness table. To see this, consider two such nodes $v_{(t)}, u_{(t')}$. If $dist_G(u,v)>1$, then we clearly have independence, since $u$ and $v$ share no dependent variables. So, assume $dist_G(u,v)\le 1$. Now, we must have $t \ne t'$, because the set of nodes resampled in any step is independent. Assume W.L.O.G that $t < t'$. Then, $v_{(t)}$ must be deeper in $T$ than $u_{(t')}$. Therefore, any shared variable of $u$ and $v$ will take a different value from the randomness table at $v_{(t)}$ than it does at $u_{(t')}$, by Lemma \ref{lem:varno}. The entries of the randomness table are sampled independently at random, and so, by independence, we have $\Prob{\text{$T$ occurs}} \le p^{|T|}$.

To bound the probability that $T$ is $R$-\emph{possible}, notice that the same argument is true for all of the nodes that are within $G$-radius $R-1$ of $root(T)$, since all of the events with which they share random variables are present in the tree, and the entries in the randomness table that cause each event to be satisfied at the step it is resampled is again uniquely defined by the tree, and disjoint. It is only the nodes $v$ with $dist_G(v, root(T))=R$ for which this is not the case, since there could be adjacent nodes $u$ with $dist_G(u, root(T))=R+1$ that were resampled in the execution log but were not added to $T$, and so the entries in the randomness table leading to $v$ being resampled are not uniquely determined by $T$. Therefore, using independence only over the events that non-boundary nodes are satisfied upon resampling, $\Prob{\text{$T$ is possible}} \le p^{|T|-|\partial(T)|}$.
\end{proof}

\subsection{Bounding the Size of Occurring Witness Trees}
Next, we wish to show that large witness trees are unlikely to occur. The next two lemmas are fairly standard, and similar arguments are the source of the bound on resampling steps in the algorithms of Moser and Tardos \cite{MT10} and Chung, Pettie, and Su \cite{CPS17}.

First, a simple upper bound on the number of witness trees of a certain size:

\begin{lemma}\label{lem:treenumber}
	The number of witness trees of size at most $k$ rooted at a node $v$ in a graph of maximum degree $d$ is less than $(5d^2)^k$.
\end{lemma}

\begin{proof}
Imagine the process of constructing a witness tree $T$ as the DFS tree produced by performing a depth-first search in $G$ (in which nodes are allowed to be explored multiple times, and are thus added as copies to $T$). We bound the number of such trees by counting the choices made when performing the depth-first search through $G$ (starting at $v$). The DFS will perform $2(k-1)$ steps, of which $k-1$ will explore a new node (potentially a copy of a node already in $T$, and $k-1$ will backtrack. If the DFS attempts to backtrack from (the original copy of) the root node $v$, we will consider this to terminate the process and return the current tree $T$ (which will be of size less than $k$) - in this way, the DFS process will be able to return all witness trees of size at most $k$, not just exactly $k$. There are at most $\binom{2(k-1)}{k-1}< 4^{k-1}$ possible sequences of exploration and backtracking steps. For each of these sequences, every exploration step has a choice of $d^2+1$ targets (all nodes within distance $2$ of the current node in $G$, including a new copy of the current node itself), while backtracking steps do not involve any further choices. So, the total number of possible DFS executions, and therefore witness trees of size at most $k$, is less than $(4d^2+4)^{k-1} \le (5d^2)^k$.
\end{proof}

We can now show an upper bound on the size of witness trees that can occur.

\begin{lemma}\label{lem:nobigtrees}
When a randomness table is sampled from its distribution, with high probability, no witness tree of size at least  $5 \log_{1/p} n $ occurs (denote this event by $E_{good}$).
\end{lemma}

\begin{proof}
By Lemma \ref{lem:treenumber}, the number of witness trees of size $k$ rooted at a particular node $v$ is less than $(5d^2)^k$. So, the total number of trees witness trees of size $k$ is less than $n(5d^2)^k$.

By Lemma \ref{lem:treeprob}, each such tree of size $k$ occurs with probability at most $p^{k}$. So, by a union bound, the probability that any witness tree of size at least $k$ occurs is bounded by:
\begin{align*}
\sum_{k=5 \log_{1/p} n}^\infty n(5d^2)^k \cdot p^{k} &\le  \sum_{k=5 \log_{1/p} n}^\infty n p^{k/2} \le 2np^{\frac 52 \log_{1/p} n}
=o(n^{-1}).
\end{align*} 

Here we use that $p\le (5d^2)^{-2}$.

\end{proof}

\subsection{Narrow Witness Trees and Risky Nodes}
Now we come to one of our most important technical definitions: narrow witness trees. Informally, these are witness trees that do not expand at their boundary:

\begin{definition}
	A witness tree $T$ of size $k$ is $\eps$-narrow if at most $\eps k$ of its nodes are boundary events.
	
\end{definition}

The definition of a node being \emph{insecure} (as in our main technical result Lemma \ref{lem:insecure}) is based on whether it is adjacent to the root of any possible narrow witness tree of a certain size:

\begin{definition}
Let $\eps\in (0,0.1)$ be a constant such that $d^{10} \le p^{2\eps-1}$ (which exists since we assume $p=d^{-(10+\Omega(1))}$), and set $\lambda = 2/\eps$. For brevity, let $\ell$ denote $\log_{\frac{1}{1-\eps}}\log_{1/p} n$.

Under a fixed randomness table, a node $v$ is \emph{risky} if it is the root of an $R$-\emph{possible} $\eps$-narrow witness tree of size at least $\lambda\ell$, for some $R\le 2(\lambda+2)\ell$.

A node $u$ is \emph{insecure} if it is either \emph{risky} or is adjacent to a risky node.
\end{definition}

We can show that, if a node is to be resampled later than some resampling step $\Omega(\log\log_{1/p} n)$ in Algorithm \ref{alg:CPSLLL}, then it must be risky:

\begin{lemma}\label{lem:mustberisky}
	For any fixed randomness table, if the event $E_{good}$ occurs (i.e. no witness tree of size larger than $5\log_{1/p} n$ occurs), and some node $v$ is in $S_t$ for $t\ge (\lambda+2)\ell$, then $v$ is risky.
\end{lemma}
\begin{proof}
	 Assume that a node $v$ is resampled in step $t$ of the execution log. For any $R$, $T_R$ be the $R$-possible witness tree, rooted at $v$, that is generated from the execution log starting at step $t$. We argue that at least one of the trees
	\[T_{\lambda\ell},T_{\lambda\ell+1},  T_{\lambda\ell+2}, \dots T_{(\lambda+2)\ell}\]
	must be $\eps$-narrow. Assume for the sake of contradiction that this is not the case. Then, we prove by induction that for $R\le (\lambda+2)\ell$, $|T_{R}|\ge (1-\eps)^{\lambda\ell-R}$.
	
	This claim is trivially true for $R\le \lambda\ell$. For the inductive step, assume the claim holds for some $R\in  [\lambda\ell, (\lambda+2)\ell)$ and consider $T_{R+1}$. The node multiset of $T_{R+1}$ must be a supermultiset of $T_{R}$ (since any node added to $T_{R}$ during the generation process would also have been added to $T_{R+1}$). By assumption, $T_{R+1}$ is not $\eps$-narrow, and therefore at least $\eps|T_{R+1}|$ of its nodes are boundary nodes. These nodes are not in $T_{R}$, since they are at distance $R+1$ (in $G$) from the common root. So, $(1-\eps)|T_{R+1}|\ge |T_{R}|$, and, by the inductive assumption,
	\[|T_{R+1}|\ge\frac{1}{1-\eps}|T_{R}| \ge \frac{1}{1-\eps} \cdot (1-\eps)^{\lambda\ell-R} \ge (1-\eps)^{\lambda\ell-(R+1)}\enspace.\]

	Now, since the event $E_{good}$ holds, no witness tree of size more than $5\log_{1/p} n$ occurs, and therefore no witness tree of size more than $5\log_{1/p} n$ is possible (since the node multiset of any possible tree is a submultiset of that of the occurring tree generated by not restricting the $G$-radius). But, we have shown that the possible tree $T_{(\lambda+2)\ell}$ must have size at least 
	\[(1-\eps)^{\lambda\ell-(\lambda+2)\ell} = (1-\eps)^{-2\ell} = \log^2_{1/p} n > 5\log_{1/p} n\enspace.\]
	
	This gives a contradiction, and therefore we conclude that at least one of $T_{\lambda\ell}\dots T_{(\lambda+2)\ell}$ must be $\eps$-narrow. This is a $2(\lambda+2)\log_{\frac{1}{1-\eps}}\log_{1/p} n$-possible $\eps$-narrow witness tree of size at least $\lambda\ell$, rooted at $v$, and so $v$ is risky.
\end{proof}

Next, we can show that narrow witness trees are unlikely to be possible, and so nodes are unlikely to be the root of any:

\begin{lemma}\label{lem:possibleprob}
	When a randomness table is sampled from its distribution, the probability that any particular node $v$ is the root of any $R$-possible $\eps$-narrow witness tree $T$ of size at least $k$ is at most $2  (5d^2p^{1-\eps})^k$, dependent on the randomness only of variables within $B_{R+1}(v)$.
\end{lemma}

\begin{proof}
By Lemma \ref{lem:treeprob}, the probability that any particular $\eps$-narrow witness tree $T$ of size at least $k$ rooted at $v$ is $R$-possible is at most $p^{(1-\eps) k}$, independently of randomness outside $B_{R+1}(v)$. By Lemma \ref{lem:treenumber} there are at most $(5d^2)^k$ trees of size $k$ rooted at $v$. So, by a union bound, the probability that $v$ is the root of any $R$-possible $\eps$-narrow witness tree $T$ of size at least $k$ is at most 

\[\sum\limits_{j=k}^\infty  p^{(1-\eps) j}(5d^2)^j \le 2 \cdot (5d^2p^{1-\eps})^k \]
	
Variables outside $B_{R+1}(v)$ are not dependent variables for any of the events contained in witness trees of $G$-radius at most $R$, so the probability holds regardless of the values of these variables.
\end{proof}

\subsection{Properties of the Graph of Insecure Nodes}

Now that we have a bound on the probability of possible narrow witness trees, we can show that nodes are unlikely to be risky, and demonstrate the properties stated in Lemma \ref{lem:insecure}:

\begin{lemma}\label{lem:riskyproperties}
Let $G'$ denote the induced graph of insecure nodes. $G'$ has the following properties, with high probability:

\begin{itemize}
	\item $G'$ is of size $np^{\Omega(\log\log_{1/p} n)}$.
	\item The connected components of $G'$ are of size $d^{O(\log\log_{1/p} n)}$ and diameter $O(\log_{1/p} n )$ (distances with respect to $G$).
	\item $G'$ admits an $(O(\log \log_{1/p} n ), O(\log^2 \log_{1/p} n ))$ network decomposition (distances again with respect to $G$).
\end{itemize} 

\end{lemma}

\begin{proof}
	Recall that a node is risky if it is the root of a $2(\lambda+2)\ell$-possible $\eps$-narrow witness trees of size at least $\lambda\ell$. The probability that a particular node $v$ is the root of any such witness tree is at most 
	$2(5d^2p^\eps)^{\lambda\ell}$, by Lemma \ref{lem:possibleprob}, and for any $u$, $v$ with $dist_G(u,v)\ge 4(\lambda+2)\ell +3$, these events are independent. 
	
	Therefore, 
	
	\[\Prob{\text{$v$ is insecure}} \le (d+1) \Prob{\text{$v$ is risky}} \le 2(d+1)(5d^2p^\eps)^{\lambda\ell}\enspace,\]
	
	and for any $u$, $v$ with $dist_G(u,v)\ge 4(\lambda+2)\ell +5$, these events are independent.

	Let $G'$ be the graph induced by insecure nodes. Let $S$ be a $(4(\lambda+2)\ell +5,4(\lambda+2)\ell +4)$-ruling set of $G'$ under distances in $G$: that is, $S$ is such that every node $v\in G'$ has $dist_G(v,S)\le 4(\lambda+2)\ell +4$, and every two nodes $u,v\in S$ have $dist_G(u,v)\ge 4(\lambda+2)\ell +5$. The existence of such a set is evident since it can be constructed greedily: considering the nodes of $G'$ in any order, and add them to the set $S$ if they are at distance at least $4(\lambda+2)\ell +5$ from any current member of $S$.

	First, we show that $S$ cannot have more than $np^{0.7\lambda\ell}$ nodes, with very high probability. Note that, since we assume $p=n^{-o(1)}$,
	
\[
np^{0.7\lambda\ell} = np^{0.7\lambda\log_{\frac{1}{1-\eps}}\log_{1/p} n}
=np^{o(\log_{1/p} n)} = n^{1-o(1)}\enspace.\]
	
Then,

\begin{align*}
	\Prob{|S|\ge np^{0.7\lambda\ell}} &\le \sum\limits_{\substack{S'\subseteq V,\\ |S'| = np^{0.7\lambda\ell}}}\Prob{S'\subseteq S}\\
	&\le \sum\limits_{\substack{S'\subseteq V,\\ |S'| = np^{0.7\lambda\ell}}}\left(2(d+1)(5d^2p^{1-\eps})^{\lambda\ell} \right)^{np^{0.7\lambda\ell}}\\
	&\le \binom{n}{np^{0.7\lambda\ell}}d^{o(\ell np^{0.7\lambda\ell})}\left(d^2p^{1-\eps}\right)^{\lambda\ell np^{0.7\lambda\ell}}\\
	&\le  d^{o(\ell np^{0.7\lambda\ell})} \left(e p^{-0.7\lambda\ell}\right)^{np^{0.7\lambda\ell}}\left(d^2p^{1-\eps}\right)^{\lambda\ell np^{0.7\lambda\ell}}&&\text{using $\binom{a}{b}\le \left(\frac {ea}{b}\right)^b$}\\
	&= d^{o(\ell np^{0.7\lambda\ell})} \left(d^2p^{0.3-\eps}\right)^{\lambda\ell np^{0.7\lambda\ell}}\\
	&= d^{o(\ell np^{0.7\lambda\ell})} \left(d^{10}p^{1.5-5\eps}\right)^{0.2\lambda\ell np^{0.7\lambda\ell}}\\
	&\le d^{o(\ell np^{0.7\lambda\ell})} \left(p^{0.5-3\eps}\right)^{0.2\lambda\ell np^{0.7\lambda\ell}}&&\text{using $d^{10}\le p^{1-2\eps}$}\\
	&\le d^{o(\ell np^{0.7\lambda\ell})} p^{0.04\lambda\ell np^{0.7\lambda\ell}}&&\text{using $\eps<0.1$}\\
	&=p^{n^{1-o(1)}}\enspace.
\end{align*}

Since every node in $G'$ is within distance $4(\lambda+2)\ell +4$ of a node in $S$, and any ball of radius $4(\lambda+2)\ell +4$ in $G$ contains at most $2d^{4(\lambda+2)\ell +4}$ nodes, with probability $1-p^{n^{1-o(1)}}$, 
\[|V(G')|\le np^{0.7\lambda\ell}\cdot 2d^{4(\lambda+2)\ell +4} \le np^{0.3\lambda\ell} = np^{\Omega(\log\log_{1/p} n)} \enspace.\]

Clearly, this also gives a bound on the number of edges in $G'$:

\[|E(G')|\le d|V(G')| = np^{\Omega(\log\log_{1/p} n)} \enspace.\]

	Next, we bound the size of the connected components of $G'$. We consider trees $T$ on the nodes of $S$ such that any $u,v$ that are adjacent in $T$ have $dist_G(u,v)\le  8(\lambda+2)\ell +9$, which we will call \emph{ruling trees}. The purpose of this is that, for any connected component $CC$ of $G'$, we can construct such a ruling tree $T$ on the nodes of $S$ that are within distance (w.r.t. $G$) $d^{4(\lambda+2)\ell +4}$ of a node in $CC$. Then, \[|V(CC)| \le |V(T)|\cdot 2d^{4(\lambda+2)\ell +4} \text{ and } diam_G(CC) \le diam_G(T)\cdot (4(\lambda+2)\ell +4)\enspace.\]

So, if we can upper-bound the size of such ruling trees $T$, we can upper-bound the size of connected components in $G'$. We can obtain such a bound by the same argument as Lemma \ref{lem:treenumber}: we consider the process of conducting a DFS to construct the ruling tree. To construct a ruling tree of size at most $k$, this DFS has at most $4^{k}$ possible sequences of exploration and backtracking steps. Each exploration step now has $2d^{8(\lambda+2)\ell  +9}$ possible destinations (since adjacent ruling tree nodes are within distance $8(\lambda+2)\ell +9$ in $G$). So, there are at most $4^{k}\cdot 2d^{8k(\lambda+2)\ell +9k}$ trees of size $k$ rooted at a node $v$, and so at most $n 4^{k}\cdot 2d^{8k(\lambda+2)\ell +9k}$ overall. The probability that each such tree's nodes are insecure (which is necessary to be in $S$, and therefore for the tree to be a ruling tree) is at most $ \left(2(d+1)\right)^{k} (5d^2p^{1-\eps})^{k\lambda\ell}$, by Lemma \ref{lem:possibleprob}. Setting $k=\frac{2}{\ell}\log_{1/p} n$,
	
	\begin{align*}
	\Prob{\exists \text{ ruling tree }T: |V(T)|\ge k, v\in V(T)} &\le n4^{k}\cdot 2d^{8k(\lambda+2)\ell +5k}\cdot \left(2(d+1)\right)^{k}(5d^2p^{1-\eps})^{k\lambda\ell} \\
	&= n^{1+o(1)} d^{8k\lambda\ell}\cdot (d^2p^{1-\eps})^{k\lambda\ell} \\
	&= n^{1+o(1)}  (d^{10}p^{1-\eps})^{k\lambda\ell} \\
	&\le n^{1+o(1)}  p^{\eps k\lambda\ell} \\
	&= n^{1+o(1)}  p^{4\log_{1/p} n} \\
	&=o(n^{-2})\enspace.
	\end{align*}
	
So, w.h.p. there is no ruling tree of size $\frac{2}{\ell}\log_{1/p} n$. Then, no connected component of $G'$ has size more than $\frac{2}{\ell}\log_{1/p} n \cdot 2d^{4(\lambda+2)\ell +2} = d^{O(\log\log_{1/p} n)}$, or diameter more than $\frac{2}{\ell}\log_{1/p} n\cdot (4(\lambda+2)\ell +4) = O(\log_{1/p} n)$.

Finally, we show the existence of a network decomposition. Contract all nodes in $G'$ to their closest ruling set node in $S$ (under distances in $G$). Connected components in this contracted graph must be subgraphs of some ruling tree, and so have size at most $\frac{2}{\ell}\log_{1/p} n$. So, there exists an $(O(\log \log_{1/p} n),O(\log \log_{1/p} n))$-network decomposition of this contracted graph, which corresponds to an $(O(\log \log_{1/p} n),O(\log^2 \log_{1/p} n ))$-network decomposition of $G'$ (where the diameter of the decomposition parts is with respect to $G$). Indeed, such a network decomposition of $G'$ can be computed deterministically in $O(\log^{O(1)} \log_{1/p} n )$ \LOCAL rounds, using the now-standard argument of Barenboim, Elkin, Pettie and Schneider \cite{BEPS16}, equipped with the subsequent network decomposition algorithm of Rozho\v{n} and Ghaffari \cite{RG20}.

\end{proof}

Now, we have all of the required properties for Lemma \ref{lem:insecure}:

\begin{proof}[Proof of Lemma \ref{lem:insecure}]
By definition, whether a node $v$ is insecure depends only on its $O(\ell)$-radius neighborhood (since insecurity is determined by whether $v$ is adjacent to the root of certain $O(\ell)$-possible witness trees, which have $O(\ell)$ $G$-radius). 

By Lemma \ref{lem:mustberisky}, if $E_{good}$ occurs (which is the case with high probability), then any node resampled after some $\Theta(\ell)$ number of resampling steps must be risky. Secure nodes are not risky and have no risky neighbors, and so must have their final dependent variable values within $O(\ell)$ steps.

By Lemma \ref{lem:riskyproperties}, we have the remaining three properties required for \ref{lem:insecure}, with high probability:

\begin{itemize}
	\item $G'$ contains $np^{\Omega(\log\log_{1/p} n)}$ nodes
	\item The connected components of $G'$ are of size $d^{O(\log\log_{1/p} n)}$ and diameter $O(\log_{1/p} n )$
	\item $G'$ admits an $(O(\log \log_{1/p} n ), O(\log^2 \log_{1/p} n ))$ network decomposition
\end{itemize} 

\end{proof}
\section{Algorithms}

Having shown Lemma \ref{lem:insecure}, we can now give a general meta-algorithm for the LLL, before showing how this meta-algorithm can be implemented to improve results in various models. The meta-algorithm is as follows (Algorithm \ref{alg:meta}):

\begin{algorithm}[H]
	\caption{LLL Meta-algorithm}
	\label{alg:meta}
	\begin{algorithmic}
		\State Sample randomness table
		\State Identify risky and insecure nodes
		\State Simulate Algorithm \ref{alg:CPSLLL} (on the whole graph) for $(\lambda+2)\ell$ resampling steps
		\State Simulate Algorithm \ref{alg:CPSLLL} for $O(\log_{1/p} n)$  resampling steps only on the induced graph of risky nodes
	\end{algorithmic}
\end{algorithm}

We detail how this meta-algorithm can be implemented in the models we study:

\subsection{\LOCAL}

In \LOCAL, the meta-algorithm of Algorithm \ref{alg:meta} can be easily implemented directly, proving Theorem \ref{thm:local}: 

\begin{proof}[Proof of Theorem \ref{thm:local}]
	
$O(1)$ rounds suffice to sample the randomness table: nodes only need to agree which node should be responsible for sampling a particular variable, which can be done by breaking ties by ID (and unique IDs in $[n^3]$ can be uniquely generated at random w.h.p.). $O(\ell) = O(\log \log_{1/p} n)$ rounds are required to identify risky and insecure nodes, and simulate Algorithm \ref{alg:CPSLLL} for $(\lambda+2)\ell$ resampling steps. Then, all secure nodes can output their current variable values and terminate, and the insecure nodes continue simulating Algorithm \ref{alg:CPSLLL} for $O(\log_{1/p} n)$ further rounds. By Lemma \ref{lem:riskyproperties}, at most $np^{\Omega(\log \log_{1/p} n)}$ nodes are insecure. So, the average number of rounds per node until termination is at most

\[O(\log \log_{1/p} n) + O(\log_{1/p} n) \cdot p^{\Omega(\log \log_{1/p} n)} = O(\log \log_{1/p} n)\enspace.\]
\end{proof}

\subsection{\lca and \vol}
In \lca and \vol, Algorithm \ref{alg:meta} can be implemented by collecting the region of the graph that can affect the output of the queried node $v$. This region consists of the $O(\log\log_{1/p} n)$-radius neighborhood of the connected component of insecure nodes containing $v$ (or just of $v$ if $v$ is not insecure).

Given a query at a node $v$, denote set $S$ to be the connected component of insecure nodes containing $v$ if $v$ is insecure, and $\{v\}$ otherwise. The algorithm is then as follows:

\begin{algorithm}[H]
	\caption{LLL in \lca and \vol, querying a node $v$}
	\label{alg:LCA}
	\begin{algorithmic}
		\State Add $v$ to a stack $T$
		\While{$T$ is nonempty}
		\State Pop $u$ from the stack
		\State Probe the (unprobed part of the) $2(\lambda+2)\ell+2$-radius neighborhood of $u$
		\State If $u$ is insecure, add to $T$ all insecure neighbors of $u$ that have not previously been added
		\EndWhile
		
		\State Simulate Algorithm \ref{alg:CPSLLL} on the probed region. Nodes at distance $i\ge 0$ from $S$ are simulated for resampling steps $t\le (\lambda+2)\ell-\lfloor i/2\rfloor$, and then risky nodes $S$ are further simulated until termination

	\end{algorithmic}
\end{algorithm}

\begin{proof}[Proof of Theorem \ref{thm:lca}]
 Algorithm \ref{alg:LCA} collects the  $O(\log\log_{1/p} n)$-radius neighborhood of $S$, and by Lemma \ref{lem:riskyproperties}, this requires $d^{O(\log\log_{1/p} n)}$ probes.  To simulate some resampling step $t$ for any node $u$, we must have simulated step $t-1$ for all nodes within distance $2$ of $u$ (to know the current variable values, and whether $u$ is a locally minimal satisfied event). We prove by induction that we have this information for all simulated nodes:
 
Base case: $t=0$. Initially, nodes within distance $2(\lambda+2)\ell$ of $S$ simulate resampling step $0$ (the initial sampling of variables). This is possible since the $2(\lambda+2)\ell+2$-radius neighborhood of $S$ has been probed.

Inductive step: for some positive $t\le (\lambda+2)\ell$, assume the claim is true for $t-1$. Then, nodes $u$ with $(\lambda+2)\ell-\lfloor dist_G(u,S)/2\rfloor \ge t$ of $S$ attempt to simulate resampling step $t$. For any such node $u$, all $w$ within $u$'s $2$-hop neighborhood have 
\[(\lambda+2)\ell-\lfloor dist_G(w,S)/2\rfloor \ge (\lambda+2)\ell-\lfloor (dist_G(u,S)+2)/2\rfloor \ge t-1,\]

So, by the inductive assumption, all such $w$ simulated resampling step $t-1$, and therefore $u$ is able to simulate resampling step $t$.

All nodes in $S$ are then able to simulate the first $(\lambda+2)\ell$ resampling steps. By \ref{lem:mustberisky}, only risky nodes resample after resampling step $(\lambda+2)\ell$. So, if $v$ is not insecure, then all of its dependent variables are fixed by resampling step $(\lambda+2)\ell$, and so it can output. Otherwise, $v$'s output variable assignment is found by simulating Algorithm \ref{alg:CPSLLL} to termination on the risky nodes in $S$. This is possible because the only neighbors of such nodes are either risky (in which case they are also participating in the simulation), or non-risky (in which case they can no longer become satisfied and need to resample, w.h.p.), and so the variable values at each resampling step are known.

\end{proof}

\subsection{\congc, linear-space \MPC, and Heterogenous \MPC}

In this section we will describe the algorithm for Heterogenous \MPC; since \congc and linear-space \MPC are strictly stronger models, the results therein follow immediately.

First, we must slightly restrict the LLL instances we study: due to \MPC's space restrictions, we cannot allow arbitrarily large numbers of variables, and arbitrarily complex events, as we can for \lca, \vol, and \local. Let $|I|$ denote the total input size. We then define constrained instances as follows:

\begin{definition}\label{def:constrained}
A constrained LLL instance is any instance such that:
\begin{itemize}
	\item Each bad event depends on $poly(d)$ variables, and (given values for those variables) can be evaluated using $poly(d)$ words of space.
	\item Given values for all variables, there exists an $O(1)$-round Heterogenous \MPC algorithm that determines all satisfied bad events, that uses $O(|I|+n^{1+o(1)})$ total space.
\end{itemize}
\end{definition}

These restrictions are fairly modest, and we are not aware of any parallel or distributed application of the LLL that does not satisfy them.

\begin{proof}[Proof of Theorem \ref{thm:mpc}]
	
We divide the analysis into two cases: $d=n^{o(1)}$, and $d=n^{\Theta(1)}$. 

In the first case, we use the following algorithm for Heterogenous \MPC (Algorithm \ref{alg:MPC}):

\begin{algorithm}[H]
	\caption{LLL in Heterogenous \MPC }
	\label{alg:MPC}
	\begin{algorithmic}
		\State Assign a machine to each node $v$
		\State $v$'s machine collects its $2(\lambda+2)\ell+2$-radius neighborhood, in $O(\log \log \log_{1/p} n)$ rounds via graph exponentiation
		\State $v$'s machine determines whether it is insecure
		\State $v$'s machine simulates $(\lambda+2)\ell$ resampling steps to determine $v$'s current variable values
		\State Collect the induced graph of risky nodes to the single large machine
		\State The large machine simulates Algorithm \ref{alg:CPSLLL} on the induced graph of risky nodes
	\end{algorithmic}
\end{algorithm}	

We collect the $O(\log \log_{1/p} n)$-radius neighborhoods necessary to determine whether nodes are risky, and to find the output values for secure nodes, onto assigned machines (note that we will assign multiple nodes to single machines, as many as the space bound allows). These neighborhoods contain $d^{O(\log \log_{1/p} n)} = d^{o(\log_{d} n)} = n^{o(1)}$ nodes and edges, and by the definition of constrained instances, a further $poly(d) = n^{o(1)}$ space is required to express each event. So, in total, $n^{o(1)}$ space is required on each machine, which fits within the space bound $\mathfrak s$. The number of rounds required to collect these neighborhoods is $O(\log \log \log_{1/p} n)$, since graph exponentiation allows neighborhoods of doubling radius to be collected each round.

This information allows machines to output for secure nodes, and only the risky nodes remain to be simulated. Since, by Lemma \ref{lem:riskyproperties}, the induced graph of insecure nodes has overall size $np^{\Omega(\log\log_{1/p} n)}$, and each node requires $poly(d)$ space for event specification, this graph the specifications of its events fits into the $O(n)$ space of the large machine. Then, this large machine can compute and output values for the remaining variables with no further communication. So, when $d=n^{o(1)}$, $O(\log \log \log_{1/p} n)$ rounds suffice to solve LLL. The total space required is $n^{1+o(1)}$, since for each node we collected a neighborhood of size $n^{o(1)}$, with $n^{o(1)}$ space per node required to specify events.

In the second case, Algorithm \ref{alg:CPSLLL} requires only $O(\log_{1/p} n) = O(1)$ resampling steps, and we can therefore afford to simulate it round-by-round. Each resampling step requires us first to identify the satisfied events, which we assume can be done in $O(1)$ rounds under the definition of constrained instances. Then, we must find the set of locally minimal satisfied events, which can be performed in $O(1)$ rounds as follows:

\begin{itemize}
	\item Restrict attention to the induced graph $G_{sat}$ of satisfied events.
	\item Each satisfied node's edges can be distributed over a contiguous block of machines, via the $O(1)$-round sorting algorithm of \cite{GSZ11}. 
	\item A machine dedicated to each node $v$ can compute $\alpha_v = \sum_{\{u,v\}\in E(G_{sat})} \textbf{1}_{ID(u)<ID(v)}$. This is a simple aggregate function, so can be computed in $O(1)$-rounds by aggregating over a $\mathfrak s$-ary tree, which will have $O(1)$ depth.
	\item If $\alpha_v = 0$, $v$ is a locally minimal satisfied event.
\end{itemize}

Finally, we resample the locally minimal satisfied events, which takes only $1$ round. Thus, when $d=n^{\Theta(1)}$, $O(1)$ rounds suffice to solve LLL. We use $O(|I|)$ space, since input information is only duplicated $O(1)$ times.

Overall, we have an $O(\log \log \log_{1/p} n)$ round Heterogenous \MPC algorithm, using $O(|I|+n^{1+o(1)})$ total space. To simulate this algorithm in linear-space \MPC, we simply have each linear-space \MPC machine simulate the behavior of $n^{1-\eps}$ low-space \MPC machines, i.e. enough to meet the higher space bound. To run the algorithm in \congc, we use the known simulation based on Lenzen's routing \cite{L13}, described in \cite{BDH18}.
\end{proof}

\section{Conclusions and Open Problems}

We have presented an analysis of resampling-style LLL algorithms that demonstrates that the outputs values for most variables can be fixed based only on a region of size $d^{O(\log\log_{1/p} n)}$ in the dependency graph. We have further shown that this improves randomized algorithmic complexities for the LLL in several distributed and models focusing on graph locality and volume.

The main open problem regarding the distributed LLL is to determine whether there exists a $\log^{O(1)}\log n$-round worst-case complexity \local algorithm for the general LLL (i.e. matching our average-case complexity for all nodes). Our work here goes some way to supporting the existence of such an algorithm: as well as our node-averaged result, the \vol model result shows that a local algorithm can solve LLL while seeing only a neighborhood of volume $d^{O(\log\log_{1/p} n)}$, which is the size of the $O(\log\log_{1/p} n)$-radius neighborhood of a node. Furthermore, if we were able to remove the need for the single large machine from the Heterogenous \MPC algorithm, and obtain a low-space \MPC algorithm, then this would also provide strong evidence that a $\log^{O(1)}\log n$-round \local algorithm should be possible: currently there are no known randomized locally-checkable problems for which the \local complexity is more than exponential in the low-space \MPC algorithm complexity (see \cite{CDP24} for discussion).

It may appear that the analysis presented here almost suffices for such a \local algorithm: the variables of almost all events can be fixed in $\log^{O(1)}\log n$ \local rounds, such that the remaining graph shatters into small components, a $(O(\log\log n),\log^{O(1)}\log n)$-network decomposition can be computed on these components, and the LLL instance is solvable without changing the fixed variables. However, the only obvious way to find output values for the remaining variables is to continue running Algorithm \ref{alg:CPSLLL}, which still takes $\Theta(\log_{1/p} n)$ rounds, and it is not clear how the network decomposition could help accelerate this.

Other open problems of interest include considering the asymmetric LLL in parallel and distributed models, determining whether efficient algorithms can attain stricter LLL criteria such as Shearer's \cite{Shearer85} rather than the polynomial criterion here, derandomizing to obtain deterministic LLL algorithms, and studying specific applications of the LLL such as edge coloring.

\hide{
	\begin{theorem}
		LLL can be solved in \LOCAL with  $O(\log\log_{1/p} n)$ node-averaged complexity.
	\end{theorem}
	
	\begin{proof}
		c
		\begin{algorithm}[H]
			\caption{LLL in \LOCAL, node-averaged complexity}
			\label{alg:LCA}
			\begin{algorithmic}
				\State Each node collects its $2(\lambda+2)\ell$-radius neighborhood and determines whether it is risky
				\State Non-risky nodes output their initial sampled variable values and terminate
				\State Risky nodes run Algorithm \ref{alg:CPSLLL}
			\end{algorithmic}
		\end{algorithm}	
		
	\end{proof}
}

 	\newcommand{\Proc}{Proceedings of the\xspace}
\newcommand{\STOC}{Annual ACM Symposium on Theory of Computing (STOC)}
\newcommand{\FOCS}{IEEE Symposium on Foundations of Computer Science (FOCS)}
\newcommand{\SODA}{Annual ACM-SIAM Symposium on Discrete Algorithms (SODA)}
\newcommand{\AISTATS}{International Conference on Artificial Intelligence and Statistics}
\newcommand{\COCOON}{Annual International Computing Combinatorics Conference (COCOON)}
\newcommand{\DISC}{International Symposium on Distributed Computing (DISC)}
\newcommand{\ESA}{Annual European Symposium on Algorithms (ESA)}
\newcommand{\ICALP}{Annual International Colloquium on Automata, Languages and Programming (ICALP)}
\newcommand{\ICML}{International Conference on Machine Learning}
\newcommand{\ICLR}{International Conference on Learning Representations}

\newcommand{\IPL}{Information Processing Letters}
\newcommand{\JACM}{Journal of the ACM}
\newcommand{\JALGORITHMS}{Journal of Algorithms}
\newcommand{\JCSS}{Journal of Computer and System Sciences}
\newcommand{\NEURIPS}{Conference on Neural Information Processing Systems}
\newcommand{\PODC}{Annual ACM Symposium on Principles of Distributed Computing (PODC)}
\newcommand{\SICOMP}{SIAM Journal on Computing}
\newcommand{\SPAA}{Annual ACM Symposium on Parallelism in Algorithms and Architectures (SPAA)}
\newcommand{\STACS}{Annual Symposium on Theoretical Aspects of Computer Science (STACS)}
\newcommand{\TALG}{ACM Transactions on Algorithms}
\newcommand{\TCS}{Theoretical Computer Science}
\bibliographystyle{plain}

\bibliography{LLL}

\end{document}